\documentclass[conference]{IEEEtran}
\IEEEoverridecommandlockouts
\usepackage{cite}
\usepackage{amsmath,amssymb,amsfonts}
\usepackage{algorithmic}
\usepackage{graphicx}
\usepackage{textcomp}
\usepackage{xcolor}
\usepackage{breqn}
\usepackage{pgfplots}
\usepackage{pgfplotstable}
\usepackage{colortbl}

\newtheorem{definition}{Definition}
\newtheorem{proof}{Proof}
\newtheorem{proposition}{Proposition}
\usepackage{subcaption}

\newcounter{eqgroup}
\newcounter{subeq}[eqgroup]
\renewcommand{\thesubeq}{C\arabic{eqgroup}.\arabic{subeq}}

\newcommand{\tagthis}{\refstepcounter{subeq}\tag{\thesubeq}}

\newcommand{\md}{\textcolor{black}}

\newcommand{\findOverallMax}[2]{%
    \pgfplotstableread[col sep=comma]{#1}\datatable

    \pgfmathsetmacro{\maxvalueA}{-1e6}
    \pgfmathsetmacro{\maxvalueB}{-1e6}
    \pgfmathsetmacro{\maxvalueC}{-1e6}
    \pgfmathsetmacro{\maxvalueD}{-1e6}

    \pgfplotstableforeachcolumnelement{Us0}\of\datatable\as\value{%
        \pgfmathsetmacro{\temp}{\value}%
        \pgfmathparse{ifthenelse(\temp>\maxvalueA,\temp,\maxvalueA)}%
        \global\let\maxvalueA\pgfmathresult%
    }

    \pgfplotstableforeachcolumnelement{Us1}\of\datatable\as\value{%
        \pgfmathsetmacro{\temp}{\value}%
        \pgfmathparse{ifthenelse(\temp>\maxvalueB,\temp,\maxvalueB)}%
        \global\let\maxvalueB\pgfmathresult%
    }

    \pgfplotstableforeachcolumnelement{Us2}\of\datatable\as\value{%
        \pgfmathsetmacro{\temp}{\value}%
        \pgfmathparse{ifthenelse(\temp>\maxvalueC,\temp,\maxvalueC)}%
        \global\let\maxvalueC\pgfmathresult%
    }

    \pgfplotstableforeachcolumnelement{Us3}\of\datatable\as\value{%
        \pgfmathsetmacro{\temp}{\value}%
        \pgfmathparse{ifthenelse(\temp>\maxvalueD,\temp,\maxvalueD)}%
        \global\let\maxvalueD\pgfmathresult%
    }

    \pgfmathsetmacro{\overallMax}{max(\maxvalueA,\maxvalueB,\maxvalueC,\maxvalueD)}

    \def#2{\overallMax}
}

\pgfmathsetmacro{\finalMaxA}{0.3}
\pgfmathsetmacro{\finalMaxB}{0.6}

\def\BibTeX{{\rm B\kern-.05em{\sc i\kern-.025em b}\kern-.08em
    T\kern-.1667em\lower.7ex\hbox{E}\kern-.125emX}}

\begin{document}

\title{Toward Energy and Location-Aware Resource Allocation in Next Generation Networks
\thanks{This work was supported by the ANR under the France 2030 program, grants "NF-NAI: ANR-22-PEFT-0003" and "NF-JEN: ANR-22-PEFT-0008"}
}

\author{Mandar Datar and Mattia Merluzzi,~\IEEEmembership{Member,~IEEE}\\
CEA-Leti, Univ. Grenoble Alpes, F-38000 Grenoble, France

}

\maketitle

\begin{abstract}
Wireless networks are evolving from radio resource providers to complex systems that also involve computing, with the latter being distributed across edge and cloud facilities. Also, their optimization is shifting more and more from a performance to a value-oriented paradigm. The two aspects shall be balanced continuously, to maximize the utilities of Services Providers (SPs), users quality of experience and fairness, while meeting global constraints in terms of energy consumption and carbon footprint among others, with all these heterogeneous resources contributing. In this paper, we tackle the problem of communication and compute resource allocation under energy constraints, with multiple SPs competing to get their preferred resource bundle by spending a a fictitious currency budget. By modeling the network as a Fisher market, we propose a low complexity solution able to achieve high utilities and guarantee energy constraints, while also promoting fairness among SPs, as compared to a social optimal solution. The market equilibrium is proved mathematically, and numerical results show the multi-dimensional trade-off between utility and energy at different locations, with communication and computation-intensive services. 

\end{abstract}

\section{Introduction}
Next generation networks will be characterized by heterogeneity from different perspectives. First, unlike previous generations, bandwidth (or, in general, wireless resources) will not be the only resource to be managed and orchestrated. In fact, computing, storage, memory and Artificial Intelligence (AI) resources will be distributed across networks, thus becoming an integrating part, as important as radio, and to be orchestrated jointly with it \cite{Dilo23}. Another heterogeneity is related to energy resources, with renewable sources distributed and possibly owned by different actors (public or private), to be also jointly considered in system optimization to minimize the negative impact of the network from an environmental perspective, but also to reduce economic costs for the operators. The energy mix is time-varying and location-dependent. For example, a distant cloud might offer low-carbon service if located in cold regions or regions with favorable energy-mix. Therefore, instantiating services requires informed decisions including the amount of resources to dedicate, their location, and ultimately the utility (or, value) they generate for the involved stakeholders (users, service providers, operators, etc.). This holistic perspective of network resources (communication and computation), energy resources, and utilities (values) will be at the focus of 6G and, in general, future information networks. However, the complexity of this management and orchestration problem should not jeopardize the benefits of this holistic approach. We need sophisticated yet simple solutions to address the problem, for instance hinging on distributed solutions that require low signaling overhead among the involved parties. This paper focuses on this complex management of heterogeneous resources in heterogeneous environments, with a Fisher Market (FM) perspective \cite{roughgarden2010algorithmic}.

\noindent \textbf{Related works.} The problem of energy in networks has been already well investigated in the literature, historically starting from energy-efficiency metrics. These metrics might be misleading due to potential rebound effects. The latter comes from the fact that higher efficiency naturally results in higher service requests, thus in increased traffic and absolute energy consumption. The focus shall shift towards absolute energy consumption-aware networking and, beyond that, carbon-aware management based on energy mix.
The authors of \cite{kamran2024green} outline ongoing standardization efforts, identify gaps in current mobile networks, and present a preliminary architectural perspective for energy-efficient, user-aware 6G networks. The work in \cite{zhong2024toward} explores the vision of carbon-neutral 6G networks, outlining key enablers such as AI-driven optimization, energy harvesting, and green network architecture design to reduce carbon emissions across the communication ecosystem. The authors of \cite{bolla20236g} propose the concept of Decarbonization Service Agreements (DSAs) to incentivize sustainable behaviors among stakeholders by integrating energy and carbon-related performance indexes into SLAs, aiming to foster a holistic green economy across 6G ecosystems. The authors in \cite{boiardi2012joint} propose a joint optimization framework for energy-aware wireless mesh networks, minimizing both capital and operational expenditures. Their results show significant energy savings over traditional two-step planning approaches, especially when allowing flexible coverage constraints. \cite{zilberman2023toward}  emphasizes the importance of carbon-aware and intelligent networking across all layers, urging for cross-disciplinary efforts to standardize carbon metrics and enable green routing decisions. 

\md{Markets, particularly FM have recently gained attention as effective tools for resource allocation that balance individual interests while satisfying key fairness and efficiency principles such as Envy-Freeness (EF), Pareto Optimality (PO), and the Sharing Incentive (SI)}\cite{buyya2008market,nguyen2019market,modina2022multi,datar2023fisher,moro2021joint}.  In our recent work \cite{YourLastName2025}, we extended the classical FM model to incorporate externalities in resource allocation, such as the environmental impact of energy consumption, by integrating the concept of Pigouvian pricing \cite{baumol1972taxation}. We adopt this extended model in our approach to support energy-aware and efficient resource allocation.

\noindent\textbf{Our contribution.} 
In this work, we consider a complex network environment where multiple Service Providers (SPs) deliver heterogeneous services, each with diverse and often conflicting resource requirements using the wireless and edge-to-cloud continuum. This includes shared access to critical resources such as radio spectrum, computing power (CPU), and memory (RAM). The network spans various geographical areas, including urban, residential, and industrial zones, each subject to different constraints on electricity consumption aimed at reducing carbon emissions. Ensuring fair and efficient resource allocation under such diverse conditions poses significant challenges. To address this, we propose a novel market equilibrium-based mechanism for joint resource allocation and pricing that balances performance, fairness, and energy towards environmental and economic sustainability.

The remainder of this paper is structured as follows: Section \ref{sys_model} presents the system model. In Sections \ref{resource_allocation_problem}, we discuss the resource allocation and pricing problem and provide the market equlibrium based solution. In Section \ref{Numerical}, we present numerical results, while Section \ref{sec:conclusions} concludes the work.

\section{System Model}\label{sys_model}
We consider a communication marketplace where a set of service providers, denoted by $\mathcal{S} = \{1, \dots, S\}$,  acquire different types of resources from an infrastructure provider (InP) to offer services to end users. These resources, indexed by the set $\mathcal{R} = \{1, \dots, R\}$, include radio spectrum, CPU, and RAM. The InP operates a network spanning multiple geographical regions, such as urban, residential, and industrial zones, collectively represented by a set $\mathcal{L}=\{1, \dots, 
L\}$, as illustrated in Fig. \ref{fig:enter-label}. 
To support services, each location $l \in \mathcal{L}$ is equipped with essential infrastructure, including base stations or cells for wireless connectivity, mobile edge computing (EC) facilities for localized computation, and cloud computing resources shared across multiple locations to handle computationally intensive tasks. Let $\mathcal{C}$ denote the set of cloud computing facilities that can be used.  
\subsection{Service utilities and costs}\label{service_utilities}

Let \( x^s_{rlc} \) represent the amount of resource \( r \) requested by service provider \( s \) at location \( l \) from cloud facility \( c \). For ease of notation, we define \( c=0 \) to represent resources obtained from the Edge Computing (EC) facility at the same location \( l \). Let \( d^s_{rlc} \),  represent the base resource demand, \emph{i.e.,} minimum amounts of resource $r$ required by SP \( s \) to achieve a unit service rate (utility) from the cloud facility \( c \) at location \( l \). Given an allocation of \( \mathbf{x}^s_{lc}=( {x}^s_{1lc}, \dots, {x}^s_{rlc} ) \), the utility of the service provider is expressed as:  
\begin{align}
    u^s_{lc}(\mathbf{x}^s_{lc})= \min_{r} \left\{ \frac{x^s_{rlc}}{d^s_{rlc}} \right\}\quad\forall l\in \mathcal{L},\forall c\in \mathcal{C}\; \label{agent_uti1}
\end{align}  
The utility function in \eqref{agent_uti1} captures a perfect complement relationship between resources. As an example, suppose a service has a base demand of $2$ units of bandwidth (e.g., resource blocks) and $1$ unit of CPU (e.g., 1 core). If it is allocated $4$ units of bandwidth and $2$ units of CPU, its service rate is:
\(
u^s_{lc} = \min\left\{ \frac{4}{2}, \frac{2}{1} \right\} = 2
\). Now, if the bandwidth is increased to $6$ units (while keeping CPU at $2$ units), the utility remains the same. Thus, utility improves only when all allocated resources increase proportionally.

\begin{figure}[t]
    \centering
    \includegraphics[width=0.75\linewidth]{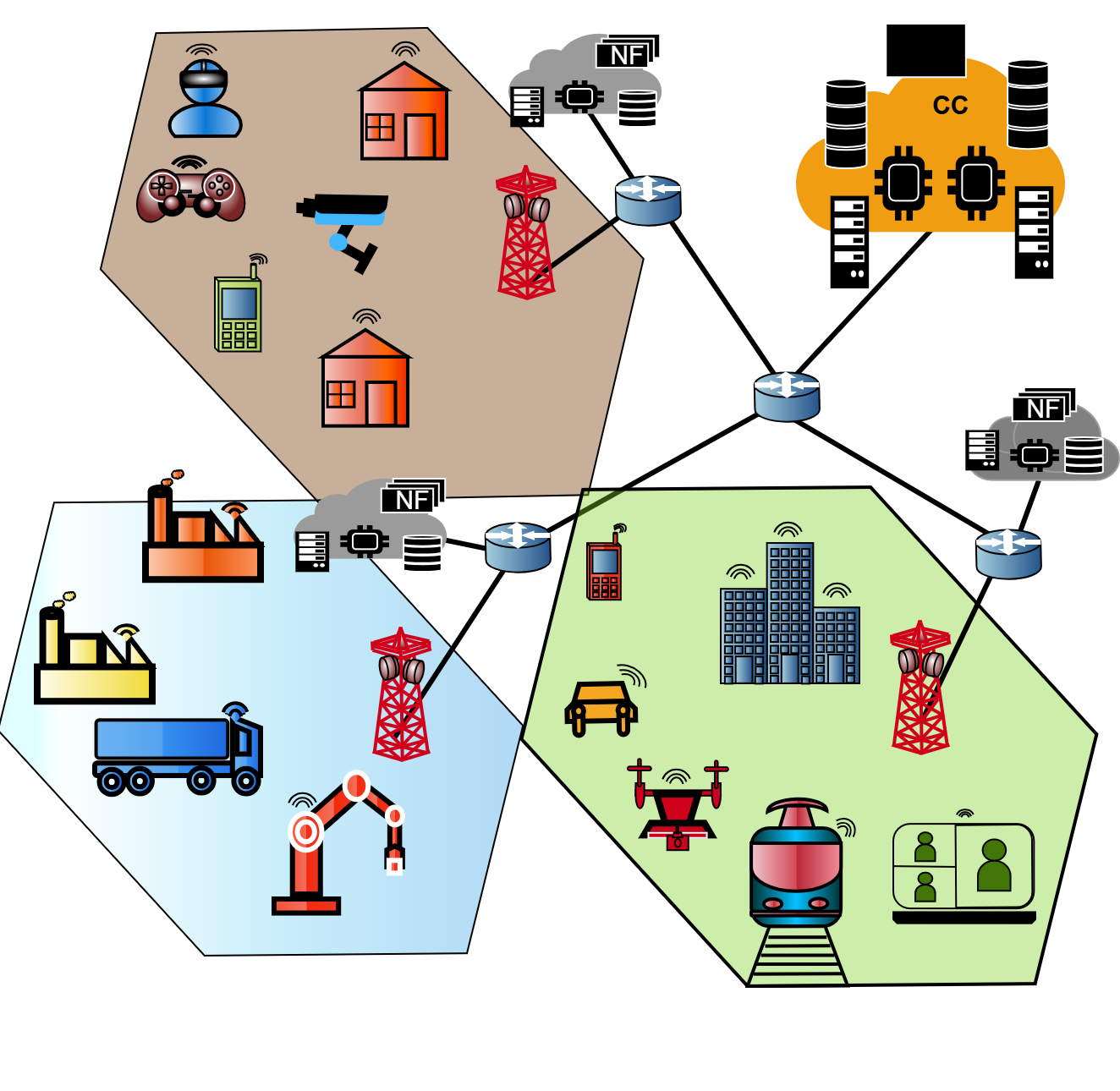}
    \captionsetup{skip=0pt}
    \caption{Network spanning diverse
geographical regions, such as urban, residential, and industrial
zones}
    \label{fig:enter-label}
\end{figure}
The base resource demand \( ( {d}^s_{1lc}, \dots, {d}^s_{rlc} ) \) varies based on the proximity of the cloud or EC facility to the service area. For example, achieving a unit service rate from a distant cloud requires more resources than from EC, depending on its actual distance from the location. The total utility obtained by service provider \( s \) to support its service  at location \( l \) from multiple cloud and EC facilities is given by:  

\begin{align}
  U^s_{l}(\mathbf{x}^s_{l}) = \sum\nolimits_{c} u^s_{lc}(\mathbf{x}^s_{lc}).\label{agent_uti2}
\end{align}  
where $\mathbf{x}^s_{l}=\left(\mathbf{x}^s_{lc}\right)_{c=0}^{c=C}$ denote the bundle of resources allocated to SP $s$ at location $l$, from different cloud facilities including the EC server located at $l$ (i.e., for $c = 0$). The overall utility achieved by service provider \( s \) across all locations is:  

\begin{align}
   U^s(\mathbf{x}^s) =\sum\nolimits_{l} U^s_{l}(\mathbf{x}^s_{l}) \label{agent_uti3}
\end{align}  
Using network resources contributes, directly or indirectly, to CO\(_2\) emissions. Since most of these emissions are due to electricity use, we focus on the energy consumption associated with resource utilization to account for the carbon footprint of the overall communication and compute networks. 
Let $\widehat{x}_{rl}$ represent the total demand for resource $r$ (including both radio and EC resources) from its local facility at location $l$. Similarly, let
\(
\widetilde{x}_{rc} = \sum_s \sum_l x^s_{rlc}
\)
represent the total demand for resource $r$ at cloud facility $c$, summed over all services $s$ and locations $l$.

We assume that total energy consumption, considering the carbon intensity of electricity generation, is subject to regulatory constraints. These constraints ensure that energy usage does not exceed specified limits set by regulatory authorities. In this work, we impose two constraints on various positions. First, we assume that each location \( l \) has a restriction on the total energy consumption. This energy consumption arises from two primary factors: the utilization of resources in the EC  facility available at that location (denoted by \( c = 0 \) ) and the energy consumed due to the total usage of radio bandwidth at that location. The latter includes the bandwidth used for uploading tasks to different cloud facilities as well as to EC facilities. Let $\widehat{e}_l(\widehat{\mathbf{x}}_{l})$ represent the total energy consumed at location $l$, and $\widetilde{e}_c(\widetilde{\mathbf{x}}_{c})$ represent the total energy consumed at cloud facility $c$. We assume that each location $l$ has a local energy consumption limit denoted by $E_l$.
\begin{equation}
    \widehat{e}_l(\widehat{\mathbf{x}}_{l}) \leq E_l \quad \forall l \in \mathcal{L}
\end{equation}
This limit depends on the characteristics of the area, such as whether it is urban, residential, or industrial. For example, urban areas may have stricter energy limits due to higher population density and concerns about emissions, while industrial areas may allow higher energy consumption to support manufacturing or heavy computing, thus balancing economic sustainability among others.
In addition to these local limits, the network is also subject to a global constraint, which caps the total energy across all locations and cloud facilities:
\begin{equation}
    \sum\nolimits_{l} \widehat{e}_l(\widehat{\mathbf{x}}_{l}) + \sum\nolimits_{c} \widetilde{e}_c(\widetilde{\mathbf{x}}_{c}) \leq E_g
\end{equation}

Each location has limited physical resources, constrained by a capacity $X_{rl}$ for each resource type $r$. In contrast, cloud resources are considered virtually unlimited. The resource usage at any location $l$ must therefore satisfy:
\begin{equation}
    \widehat{x}_{rl} \leq X_{rl} \quad \forall l \in \mathcal{L}
\end{equation}

\section{Resource allocation and pricing problem}\label{resource_allocation_problem}
We consider that each service provider acts selfishly, aiming to maximize its individual utility. Service providers do not necessarily account for the impact of their resource usage on the collective emissions caused by other SPs. This behaviour may lead to violations of the regulatory limits on CO$_2$ emissions. The objective of this work is to design a pricing mechanism for resource usage that penalizes SPs to maintain CO$_2$ emissions below the prescribed thresholds. However, when applying the mechanism, our goal is to harmonize the needs of all service providers and ensure a fair allocation. \md{Toward this goal, we propose a resource allocation scheme based on the Fisher market model, formally defined as 
$$\mathcal{M}:=\!\left \langle \mathcal{S},\left(\mathbf{x}^{s}\in\mathbb{R}^{K}\right)_{s\in\mathcal{S}},\left(U^{s}\right)_{s\in\mathcal{S}},\left(B^{s}\right)_{s\in\mathcal{S}},\mathbf{p}\in\mathbb{R}^K \!\right \rangle, $$ 
where a set of buyers \( \mathcal{S} \) with budgets \( (B^s)_{s \in \mathcal{S}} \) compete to purchase bundles \( \mathbf{x}^s \in \mathbb{R}^K \) of divisible goods, priced by \( \mathbf{p} \in \mathbb{R}^K \), to maximize their individual utilities \( (U^s)_{s \in \mathcal{S}} \), where $K=R\times L\times (C+1)$ in our case.
The market seeks an equilibrium where each buyer spends their budget optimally and the market clears.} In our setting,  each \( s \in \mathcal{S} \) service provider (\emph{i.e.,} buyer) is allocated a budget \( B^s \) representing artificial currency or carbon credits. This budget, which is set by the regulator, may depend on factors such as the service provider's priority or its efforts to reduce carbon emissions, for example, by adopting renewable energy sources. This allocated budget allows SPs to procure resources from different locations based on the services they provide. In this setup, a market operator sets the prices for each resource, and SPs are required to pay these prices (including taxes) to use them. Let \( p_{rlc} \) represent the price per unit for using resource $r$ at location $l$ of cloud facility $c$ and the vector $\mathbf{p}=\left( p_{rlc}\right)_{r\in\mathcal{R},c\in\mathcal{C},l\in \in \mathcal{L}}$ denote the collective prices.

  

\subsection{Service providers problem}\label{sec:SP_problem}
We consider that the SPs are rational and given the announced prices, they allocate their budgets across different resources in different locations to procure the optimal bundle of heterogeneous resources required to support their services. Thus, the decision problem for each SP is defined as follows:
\[
\begin{aligned}
Q_s: \quad & \underset{\mathbf{x}^s}{\text{maximize}} & &  \sum_{l}\sum_{c}u^s_{lc}(\mathbf{x}^s_{lc}) \\
& \text{subject to} & & 
\sum\nolimits_{r,l,c} p_{rlc}  x^s_{rlc} \leq B_s, \\
&&& \sum\nolimits_{c}u^s_{lc}\leq \widehat{U}^s_l \quad \forall l \in \mathcal{L}\\
& & & x^s_{rlc}\geq 0 \;\forall c\in\mathcal{C},l\in\mathcal{L}.
\end{aligned}
\]
with $u^s_{lc}(\mathbf{x}^s_{lc})= \min_{r} \left\{ \frac{x^s_{rlc}}{\overline{d}^s_{rlc}}\right\}, \forall c\in\mathcal{C},l\in\mathcal{L}$.
In the above program, note that the SPs’ decisions (\emph{i.e, resource demand}) depend on the prices set for the resources. However, given any prices, there is no guarantee that the total demand for each resource will remain within the resource capacity or that the energy consumption due to its usage will remain within the allowed limits. To address this, the market operator adjusts the prices to ensure that resource usage and energy consumption stay within their respective limits. In economic theory, such a balanced state is known as a market equilibrium (ME) or competitive equilibrium (CE) \cite{arrow1954existence}.
\subsection{Market Equilibrium }
\begin{definition}
A market equilibrium (ME) for the market $\mathcal{M}$ is defined as a pair $(\mathbf{p}^*,\mathbf{x}^{*})$ of prices and allocation, where the market meets its total energy restrictions and service providers get their preferred resource bundle. Mathematically $(\mathbf{p}^*,\mathbf{x}^{*})$ is a CE if the following two conditions are satisfied.\\
\textbf{C1} Given the price vector, every SP $s$ spends its budget such that it receives resource bundle ${\mathbf{x}^{s}}^{*}$ that maximizes its utility. 
\begin{equation}
\scalebox{0.93}{$
\displaystyle
{\mathbf{x}^{s}}^{*} \in \arg\max \left \{ U^{s}(\mathbf{x}^s) \,\middle|\, \sum_{r,l,c} p^*_{rlc}  x^s_{rlc} \leq B_s \right \}
\quad \forall s \in \mathcal{S}
$}
\tag{C1.1} \label{bestdemand}
\end{equation}
\stepcounter{eqgroup}
\textbf{C2} If the total energy consumption due to resources usage meets the capacity, it is positively priced; otherwise, the corresponding resource  has zero price, i.e., we have:
\begin{align}
    &\scalebox{0.96}{ $p^*_{rlc}=\lambda\nabla_{ \widetilde{x}_{rc}}\widetilde{e}_{c}(\widetilde{x}^*_{c})\quad \forall r\in \mathcal{R},\; \forall l\in \mathcal{L},\;\forall c\in \mathcal{C} \setminus \{0\},$}\tagthis \label{c2.1} \\
    &\scalebox{0.96}{$ p^*_{rl0}=\gamma_{rl}+(\lambda+\mu_{l})\nabla_{ \widehat{x}_{rc}}\widehat{e}_{\ell}(\widehat{x}^*_{l})\quad \;\forall r\in \mathcal{R},\; \forall l\in \mathcal{L},$}\tagthis \label{c2.2}\\
     &\scalebox{0.96}{ $\gamma_{rl}(\widehat{x}^*_{rl}- X_{rl})=0\quad\forall r\in \mathcal{R},\; \forall l\in \mathcal{L},$}\tagthis \label{c2.3}\\
    &\scalebox{0.96}{$\mu_{l}\left(\widehat{e}^l(\mathbf{x}^*)-E_{\ell}\right)=0\quad \forall l\in \mathcal{L}$} \tagthis \label{c2.4}\\
&\scalebox{0.96}{$\lambda\left(\sum_{l}\widehat{e}_l(\widehat{\mathbf{x}}^*_{l})+\sum_{c}\widetilde{e}_c(\widetilde{\mathbf{x}}^*_{c})- E_{g} \right)=0 $}\tagthis \label{c2.5}
\end{align}
In the above conditions, $\gamma_{r\ell c}$ represents the actual price per unit of resource $r$, $\mu_\ell$ is an additional tax applied at location $\ell$ due to local energy consumption limits, and $\lambda$ is a global tax imposed to keep overall resource usage below a threshold.

\end{definition}
In the next section, we demonstrate that the market equilibrium solution to the formulated market can be computed by solving the convex optimization program.   
\subsection{Solution}\label{solution}

\begin{proposition}
 Consider a market where each agent’s utility is defined as in \eqref{agent_uti1}–\eqref{agent_uti3}. Suppose that the functions $\widehat{e}_l(\widehat{\mathbf{x}}_{l})$ and $\widetilde{e}_c(\widetilde{\mathbf{x}}_{c})$ are convex and increasing in $\widehat{\mathbf{x}}_{l}$ and $\widetilde{\mathbf{x}}_{c}$, respectively, for all $l \in \mathcal{L}$ and $c \in \mathcal{C}$. Then, the market equilibrium (ME) can be obtained by solving the optimization problem \eqref{EG_Program}. Also, the optimal allocation $\mathbf{x}^*$ and the associated dual variables $\mathbf{p}^*$ (corresponding to constraints \eqref{c2.1}–\eqref{c2.5}) together constitute a market equilibrium.
\begin{equation}
\begin{aligned}
& \underset{\mathbf{x}}{\text{maximize}} 
& & \sum_s B^s\log\left(\sum_{l}\sum_{c}u^s_{lc}(\mathbf{x}^s_{lc})\right) \\
& \text{subject to} 
& & u^s_{lc}(\mathbf{x}^s_{lc}) = \min_{r} \left\{ \frac{x^s_{rlc}}{\overline{d}^s_{rlc}} \right\}, \forall s, \forall l\in\mathcal{L}, \forall c\in\mathcal{C}, \\
& (\gamma_{rl}) 
& & \widehat{x}_{rl} - X_{rl} \leq 0, \forall r\in\mathcal{R}, \forall l\in\mathcal{L}, \\
& (\mu_l) 
& & \widehat{e}_l(\widehat{\mathbf{x}}_l) \leq E_l, \forall l\in\mathcal{L}, \\
& (\lambda) 
& & \sum_{l}\widehat{e}_l(\widehat{\mathbf{x}}_l) + \sum_{c}\widetilde{e}_c(\widetilde{\mathbf{x}}_c) \leq E_g, \\
& 
& & x^s_{rlc} \geq 0, \forall s, \forall r\in\mathcal{R}, \forall l\in\mathcal{L}, \forall c\in\mathcal{C}.
\end{aligned}
\label{EG_Program}
\end{equation}
\end{proposition}
\begin{proof}
    As the utility function of each SP (agent) is concave homogeneous of degree one, and $\widehat{e}_l(\widehat{\mathbf{x}}_{l})$ and $\widetilde{e}_c(\widetilde{\mathbf{x}}_{c})$ are convex and increasing in $\widehat{\mathbf{x}}_{l}$ and $\widetilde{\mathbf{x}}_{c}$, the claim directly follows from \cite[Theorem 1]{YourLastName2025}.
\end{proof}
We have shown that ME can be found by solving the optimization program \eqref{EG_Program}.
This program's objective is equivalent to maximizing Nash social welfare.
As a result, the equilibrium allocation also maximizes the Nash welfare or ensures proportional fairness while distributing resources among the SPs.
\section{Numerical Simulations}\label{Numerical}
For our numerical experiments, we examine a scenario involving four competing SPs operating across two distinct locations. Each SP offers a unique type of application, distinguished by its primary resource demand: one is radio-intensive, another is computing-intensive, the third requires high memory (RAM), and the fourth has balanced resource needs.
\renewcommand{\arraystretch}{0.95}
\begin{table}[h]
\centering
\caption{The base demand vector of service classes}
\begin{tabular}{|c|c|c|c|}
\hline 
User Applications & VCPU  & RAM (GB) & NB (Mbps) \\ 
\hline 
BW-Intensive & 2-4 & 8-12 & 300-492 \\ 
\hline 
CPU-Intensive & 30-36 & 6-8 & 50-70 \\ 
\hline 
RAM-Intensive & 2-4  & 28-32 & 50-70 \\ 
\hline 
Balanced & 2-4 & 3.5-4 & 50-70 \\ 
\hline 
\end{tabular} 
\label{basedemandtable}
\end{table}
 The exact SPs' resource requirements are reported in the Table. Each location has wireless communication resources and an EC facility with CPUs and RAM. Both locations are also connected to two cloud servers that provide additional computing power, including CPUs and RAM. We neglect the energy cost of the RAM, against radio and CPU usage. For radio resource as it is available at location and is used to connect EC and cloud \( \widehat{x}_{\text{\tiny RAN},l}=\sum_s \sum_c x^s_{\text{\tiny RAN},lc}\) while \( \widehat{x}_{\text{\tiny CPU},l} = \sum_s x^s_{\text{\tiny CPU},l0} \) and \( \widehat{x}_{\text{\tiny RAM},l} = \sum_s x^s_{\text{\tiny RAM},l0} \) denote the total CPU and RAM demand, respectively, at the EC facility located at \( l \).
Similarly, let \( \widetilde{x}_{\text{\tiny CPU},c} = \sum_s \sum_l x^s_{\text{\tiny CPU},lc} \) and \( \widetilde{x}_{\text{\tiny CPU},c} = \sum_s \sum_l x^s_{\text{\tiny RAM},lc} \) denote the total demand for resource CPU and RAM at cloud facility \( c \). The total energy consumption due to  RAN and CPU at each location $l$ is given as  
\[
\widehat{e}_l(\widehat{\mathbf{x}}_l)=\widehat{e}_{\text{\tiny RAN},l}\left(\widehat{x}_{\text{\tiny RAN},l}\right)+ \widehat{e}_{\text{\tiny CPU},l}\left(\widehat{x}_{\text{\tiny CPU},l}\right)\;\forall l \in \mathcal{L}, \label{constraint1} 
\]
and the total energy consumption in the clouds is given as 
\[
\widetilde{e}_c(\widetilde{\mathbf{x}}_c)= \widetilde{e}_{\text{\tiny CPU},c}\left(\widetilde{x}_{\text{\tiny CPU},c}\right)\;\forall c \in \mathcal{C}, \label{constraint1} 
\]
 where we consider that the energy consumption for each RAN and CPU resource follows a power function: \( e_k(\overline{x}_k) = (\overline{x}_k)^{\beta_k} \), where \( \beta_k \geq 1 \). This formulation helps to analyze different trade-offs in energy consumption. Also, in case of CO$_2$ constraints, it is an abstract way of representing different energy mix at different locations. In this paper, we do not go into the details of the specific functions, but rather explore the effects of $\beta_k$ on the overall trade-off between utility and costs (e.g., energy and carbon footprint). We impose energy constraints at two levels:
\textit{(i)} overall energy consumption across all resources and
\textit{(ii)} local energy limits for radio resources and the EC facility at each location. We apply the proposed resource allocation method by solving equation \eqref{EG_Program}, and compare the resulting ME allocation with the social optimum (SO), which maximizes the total weighted service provided by SPs while meeting energy and resource capacity constraints.
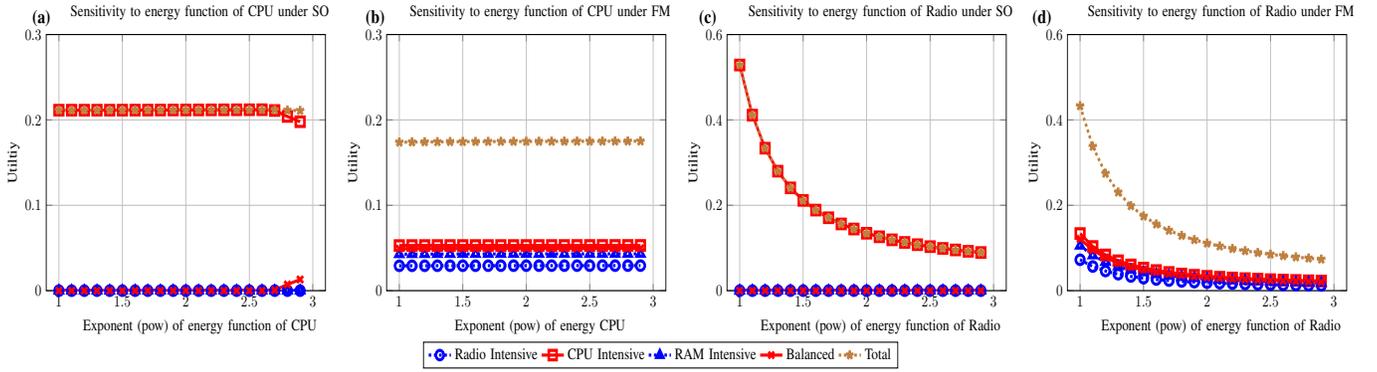
\begin{figure*}
    \resizebox{\textwidth}{4.9cm}{
   \begin{tikzpicture}
\node[anchor=south west] at (-2,5.75){\large\textbf{(a)}};
   \node[anchor=south west] at (7,5.75){\large\textbf{(b)}};
    \node[anchor=south west] at (16,5.75){\large\textbf{(c)}};
     \node[anchor=south west] at (25,5.75){\large\textbf{(d)}};
        \begin{axis}[
             name=ax1,
            title={Sensitivity to energy function of CPU under SO },
            xshift=-1.5cm,
            xmin=0.9, xmax=3.1,
            ymin=0, ymax=\finalMaxA,
            xscale=1.1, yscale=1,
            xlabel={Exponent (pow) of energy function of CPU},
            ylabel={Utiltiy},
            ylabel style={yshift=-10pt},
            grid=major,
            legend style={at={(2,-0.2)}, anchor=north,legend columns=-1},
        ]

           \addlegendentry{Radio Intensive}
             \addplot[dotted, line width=2pt,mark=o,mark options={solid},mark size=3,blue] table [x index=0, y index=1, col sep=comma] {Pictures/filtered_1_multi.csv};
               \addlegendentry{CPU Intensive}
             \addplot[line width=2pt,mark=square,mark size=3, red] table [x index=0, y index=2, col sep=comma] {Pictures/filtered_1_multi.csv};
             \addlegendentry{RAM Intensive}
             \addplot[dotted,line width=2pt,mark=triangle,mark options={solid},mark size=3, blue] table [x index=0, y index=3, col sep=comma] {Pictures/filtered_1_multi.csv};
            \addlegendentry{Balanced} 
             \addplot[line width=2pt,mark=x,mark size=3,red] table [x index=0, y index=4, col sep=comma] {Pictures/filtered_1_multi.csv};
             \addlegendentry{Total}
              \addplot [dotted,line width=2pt,mark=star,mark options={solid}, mark size=3,brown] table [x index=0, 
y expr=\thisrow{Us0} + \thisrow{Us1}+ \thisrow{Us2}+ \thisrow{Us3}, col sep=comma] {Pictures/filtered_1_multi.csv};
\end{axis}

\begin{axis}[
             name=ax2,
        at={(ax1.south east)},
         xmin=0.9, xmax=3.1,
          ymin=0, ymax=\finalMaxA,
        xscale=1.1, yscale=1,
        xshift=1.5cm,
            title={Sensitivity to energy function of CPU under FM },
            xlabel={Exponent (pow) of energy CPU},
            ylabel={Utility},
            ylabel style={yshift=-10pt},
            grid=major,
        ]

            \addplot[dotted, line width=2pt,mark=o,mark options={solid},mark size=3,blue] table [x index=0, y index=1, col sep=comma] {Pictures/filtered_2_multi.csv};
              
             \addplot[line width=2pt,mark=square,mark size=3, red] table [x index=0, y index=2, col sep=comma] {Pictures/filtered_2_multi.csv};
           
             \addplot[dotted,line width=2pt,mark=triangle,mark options={solid},mark size=3, blue] table [x index=0, y index=3, col sep=comma] {Pictures/filtered_2_multi.csv};
          
             \addplot[line width=2pt,mark=x,mark size=3,red] table [x index=0, y index=4, col sep=comma] {Pictures/filtered_2_multi.csv};
             \addplot[dotted,line width=2pt,mark=star,mark options={solid}, mark size=3,brown] table [x index=0, 
y expr=\thisrow{Us0} + \thisrow{Us1}+ \thisrow{Us2}+ \thisrow{Us3}, col sep=comma] {Pictures/filtered_2_multi.csv};
\end{axis}

\begin{axis}[
             name=ax3,
        at={(ax2.south east)},
         xmin=0.9, xmax=3.1,
          ymin=0, ymax=\finalMaxB,
       xscale=1.1, yscale=1,
        xshift=1.5cm,
            title={Sensitivity to energy function of Radio under SO},
            xlabel={Exponent (pow) of energy function of Radio},
            ylabel={Utility},
            ylabel style={yshift=-10pt},
            grid=major,
        ]

            \addplot[dotted, line width=2pt,mark=o,mark options={solid},mark size=3,blue] table [x index=0, y index=1, col sep=comma] {Pictures/filtered_3_multi.csv};
              
             \addplot[line width=2pt,mark=square,mark size=3, red] table [x index=0, y index=2, col sep=comma] {Pictures/filtered_3_multi.csv};
           
             \addplot[dotted,line width=2pt,mark=triangle,mark options={solid},mark size=3, blue] table [x index=0, y index=3, col sep=comma] {Pictures/filtered_3_multi.csv};
          
             \addplot[line width=2pt,mark=x,mark size=3,red] table [x index=0, y index=4, col sep=comma] {Pictures/filtered_3_multi.csv};

             \addplot[dotted,line width=2pt,mark=star,mark options={solid}, mark size=3,brown] table [x index=0, 
y expr=\thisrow{Us0} + \thisrow{Us1}+ \thisrow{Us2}+ \thisrow{Us3}, col sep=comma] {Pictures/filtered_3_multi.csv};
\end{axis}

\begin{axis}[
             name=ax4,
        at={(ax3.south east)},
         xmin=0.9, xmax=3.1,
          ymin=0, ymax=\finalMaxB,
       xscale=1.1, yscale=1,
        xshift=1.5cm,
            title={Sensitivity to energy function of Radio under FM},
            xlabel={Exponent (pow) of energy function of Radio },
            ylabel={Utility},
            ylabel style={yshift=-10pt},
            grid=major,
        ]

             \addplot[dotted, line width=2pt,mark=o,mark options={solid},mark size=3,blue] table [x index=0, y index=1, col sep=comma] {Pictures/filtered_4_multi.csv};
              
             \addplot[line width=2pt,mark=square,mark size=3, red] table [x index=0, y index=2, col sep=comma] {Pictures/filtered_4_multi.csv};
           
             \addplot[dotted,line width=2pt,mark=triangle,mark options={solid},mark size=3, blue] table [x index=0, y index=3, col sep=comma] {Pictures/filtered_4_multi.csv};
          
             \addplot[line width=2pt,mark=x,mark size=3,red] table [x index=0, y index=4, col sep=comma] {Pictures/filtered_4_multi.csv};

             \addplot[dotted,line width=2pt,mark=star, mark options={solid}, mark size=3,brown] table [x index=0, 
y expr=\thisrow{Us0} + \thisrow{Us1}+ \thisrow{Us2}+ \thisrow{Us3}, col sep=comma] {Pictures/filtered_4_multi.csv};
\end{axis}

\end{tikzpicture}

}

\caption{ Impact of different energy functions on agents' utility. (a) and (b): Sensitivity to the CPU energy function under Social Optimum (SO) and Fisher Market (FM)-based allocation, respectively. (c) and (d): Sensitivity to the Radio energy function under SO and FM-based allocation, respectively. }
    \label{fig:sensitivity}
\end{figure*}
\par In Figs. \ref{fig:sensitivity} we analyze the sensitivity of cost functions on SPs' utilities under both allocation schemes,  
and they are obtained by adjusting the power of the cost function $\beta_k$ from 1 to 3. \md{ From Fig. \ref{fig:sensitivity}(a)-(b), we observe that changing the power of the CPU energy cost has little effect on the overall utilities of the SPs. This is because only restrictions at locations level are active, while the global constraint is inactive. Since SPs use CPUs in the cloud, varying the CPU energy function affects utility only when the global constraint activates, around exponent 2.5–3 in Fig. 2(a). In contrast, utilities are more sensitive to changes in the radio energy function due to the active location level restrictions.} However, in Fig. \ref{fig:sensitivity}(a), which shows the SO scheme, the utilities are unevenly distributed—some SPs benefit significantly while others receive very little. In contrast, Fig. \ref{fig:sensitivity}(b) shows that the FM scheme leads to a fairer allocation among SPs. Despite the difference in fairness, the total utilities (efficiency) under both schemes remain similar. Figs.  \ref{fig:sensitivity}(c)-(d) show that the SPs' utilities decrease as the power of the radio energy function increases. We also observe that the FM scheme provides better fairness compared to the SO scheme as in previous case. 
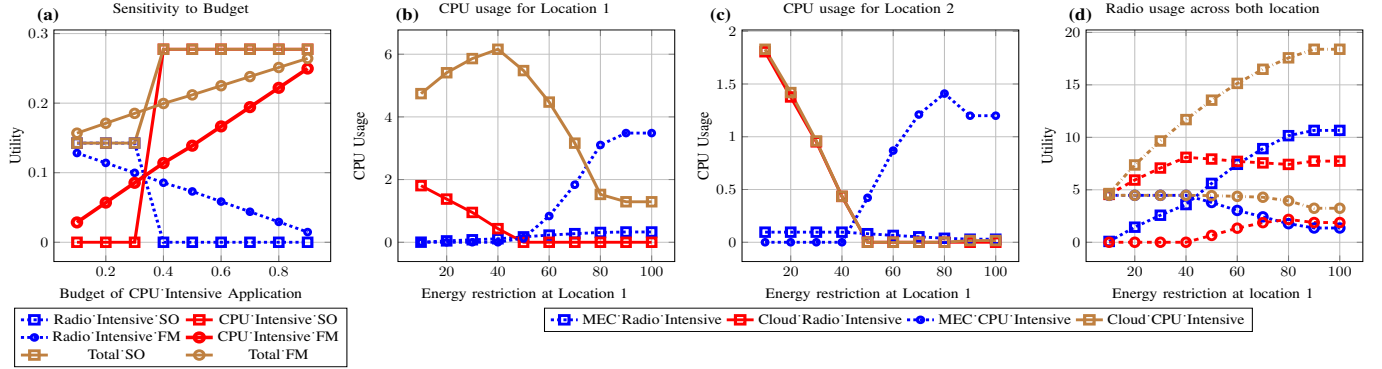
\begin{figure*}
    \resizebox{\textwidth}{4.9cm}{
   \begin{tikzpicture}
   \node[anchor=south west] at (1,5.75){\large\textbf{(a)}};
   \node[anchor=south west] at (9,5.75){\large\textbf{(b)}};
    \node[anchor=south west] at (16,5.75){\large\textbf{(c)}};
     \node[anchor=south west] at (24,5.75){\large\textbf{(d)}};
\begin{axis}[xshift=1.5cm,
            title={Sensitivity to Budget},
            name=ax1,
            xscale=0.9, yscale=1,
            xlabel={Budget of CPU_Intensive Application},
            ylabel={Utility},
            ylabel style={yshift=-10pt},
            grid=major,
            legend style={at={(0.5,-0.2)}, anchor=north,legend columns=2},
        ]
           \addlegendentry{Radio_Intensive_SO}
             \addplot [dotted,line width=2pt,mark=square,mark options={solid},mark size=3,blue] table [x index=0, y index=1, col sep=comma] {Pictures/filtered_1_twoSP.csv};
              \addlegendentry{CPU_Intensive_SO}
             \addplot [line width=2pt,mark=square, mark size=3,red] table [x index=0, y index=2, col sep=comma] {Pictures/filtered_1_twoSP.csv};
             \addlegendentry{Radio_Intensive_FM}
             \addplot[dotted,line width=2pt,mark=o,
    mark options={solid}, blue] table [x index=0, y index=1, col sep=comma] {Pictures/filtered_2_twoSP.csv};
    \addlegendentry{CPU_Intensive_FM} 
             \addplot[line width=2.5pt,mark=o,mark size=3, red] table [x index=0, y index=2, col sep=comma] {Pictures/filtered_2_twoSP.csv};
\addlegendentry{Total_SO}
     \addplot[line width=2pt,mark=square,mark size=3, brown] table [x index=0,  y expr=\thisrow{Us0} + \thisrow{Us1}, col sep=comma] {Pictures/filtered_1_twoSP.csv};
\addlegendentry{Total_FM}
\addplot[line width=2pt,mark=o,mark size=3, brown] table [x index=0,  y expr=\thisrow{Us0} + \thisrow{Us1}, col sep=comma] {Pictures/filtered_2_twoSP.csv};

        \end{axis}
\begin{axis}[
         at={(ax1.south east)}, xshift=1.5cm,
         	name=ax2,
            title={CPU usage for Location 1},
            xscale=0.9, yscale=1,
            xlabel={Energy restriction at Location 1},
            ylabel={CPU Usage},
            ylabel style={yshift=-10pt},
            grid=major,
            legend style={at={(2,-0.2)}, anchor=north,legend columns=4},
        ]
            \addlegendentry{MEC_Radio_Intensive}
             \addplot [dotted,line width=2pt,mark=square,mark options={solid},mark size=3,blue] table [x index=0, y index=1, col sep=comma] {Pictures/filtered_4_twoSP.csv};
              \addlegendentry{Cloud_Radio_Intensive}
             \addplot [line width=2pt,mark=square, mark size=3,red] table [x index=0, 
y expr=\thisrow{S0L0K1R1} + \thisrow{S0L0K2R1}, col sep=comma] {Pictures/filtered_4_twoSP.csv};
 				\addlegendentry{MEC_CPU_Intensive}             
             \addplot[dotted,line width=2pt,mark=o,
    mark options={solid}, blue] table [x index=0, y index=3, col sep=comma] {Pictures/filtered_4_twoSP.csv};
            \addlegendentry{Cloud_CPU_Intensive} 
\addplot[line width=2pt,mark=square,mark size=3, brown] table [x index=0,  y expr=\thisrow{S1L0K1R1} + \thisrow{S1L0K2R1}, col sep=comma] {Pictures/filtered_4_twoSP.csv};

 \end{axis}
\begin{axis}[
         at={(ax2.south east)}, xshift=1.5cm,
         	name=ax3,
            xscale=0.9, yscale=1,
            title={CPU usage for Location 2},
            xlabel={Energy restriction at Location 1},
            ylabel={CPU Usage},
            ylabel style={yshift=-10pt},
            grid=major,
        ]
             \addplot [dotted,line width=2pt,mark=square,mark options={solid},mark size=3,blue] table [x index=0, y index=2, col sep=comma] {Pictures/filtered_4_twoSP.csv};
             \addplot [line width=2pt,mark=square, mark size=3,red] table [x index=0, 
y expr=\thisrow{S0L1K1R1} + \thisrow{S0L1K2R1}, col sep=comma] {Pictures/filtered_4_twoSP.csv};
             \addplot[dotted,line width=2pt,mark=o,
    mark options={solid}, blue] table [x index=0, y index=4, col sep=comma] {Pictures/filtered_4_twoSP.csv};
            
\addplot[line width=2pt,mark=square,mark size=3, brown] table [x index=0,  y expr=\thisrow{S1L1K1R1} + \thisrow{S1L1K2R1}, col sep=comma] {Pictures/filtered_4_twoSP.csv};

 \end{axis} 
 
\begin{axis}[
         at={(ax3.south east)},xshift=1.5cm,
         	name=ax4,
            title={Radio usage across both location},
            xscale=0.9, yscale=1,
            xlabel={Energy restriction at location 1},
            ylabel={Utility},
            ylabel style={yshift=-10pt},
            grid=major,
            legend style={at={(6,-0.2)}, anchor=north,legend columns=1},
        ]
             \addplot [dotted,line width=2pt,mark=square,mark options={solid},mark size=3,blue] table [x index=0, y index=3, col sep=comma] {Pictures/filtered_5_twoSP.csv};
             
             \addplot [dashed,line width=2pt,mark=square,mark options={solid},mark size=3,red] table [x index=0, y index=4, col sep=comma] {Pictures/filtered_5_twoSP.csv};

\addplot [dotted,line width=2pt,mark=o,mark options={solid},mark size=3,blue] table [x index=0, y index=5, col sep=comma] {Pictures/filtered_5_twoSP.csv};

\addplot [dashed,line width=2pt,mark=o,mark options={solid},mark size=3,red] table [x index=0, y index=6, col sep=comma] {Pictures/filtered_5_twoSP.csv};

\addplot [dash dot,line width=2pt,mark=square,mark options={solid},mark size=3,brown] table [x index=0, y index=1, col sep=comma] {Pictures/filtered_5_twoSP.csv};

\addplot [dash dot,line width=2pt,mark=o,mark options={solid},mark size=3,brown] table [x index=0, y index=2, col sep=comma] {Pictures/filtered_5_twoSP.csv};

 \end{axis} 
 
\end{tikzpicture}}

\caption{Variation in the usage of: (a) radio, (b) CPU, and the utility derived from both MEC and cloud facilities vs. local energy consumption constraints. }
    \label{fig:restriction}
\end{figure*}
To simplify the analysis, in the remaining simulations, shown in Figs. \ref{fig:restriction} and  \ref{fig:multiple_images}, we focus on two SPs: the first one offering a radio-intensive service, and the other a CPU-intensive service. In Fig. \ref{fig:restriction}(a), we study how resource allocation is influenced by the SPs’ budgets. Specifically, we vary the budget of the radio-intensive SP from 0.1 to 0.9 and observe the resulting utilities for both SPs under two allocation schemes. The results show that under the SO scheme, the SP with a lower budget receives significantly less utility. In contrast, the proposed ME-based approach is less sensitive to budget variations, and strikes a more balanced trade-off between cost, utility, and fairness. \par Fig. \ref{fig:restriction} (b), (c), and (d) show how local energy constraints affect resource allocation and SPs utilities. For simulations, we vary the energy limit at location 1 from 10 to 90 units, while keeping it fixed at other locations. Fig. \ref{fig:restriction}(b) shows how CPU usage at EC location 1 and the cloud facilities changes under different restriction levels. For the CPU-Intensive case, as restrictions are gradually relaxed, cloud usage initially increases since more radio resources become available for offloading. However, with further relaxation, EC usage starts to increase, reducing dependency on the cloud. In contrast, for the Radio-Intensive case, the cloud is initially used more intensively. However, as restrictions ease, cloud usage decreases slightly while EC usage increases, though not as significantly as in the CPU-Intensive case. 
\begin{figure*}[ht] \centering \begin{subfigure}{0.42\textwidth} \centering \includegraphics[width=\linewidth]{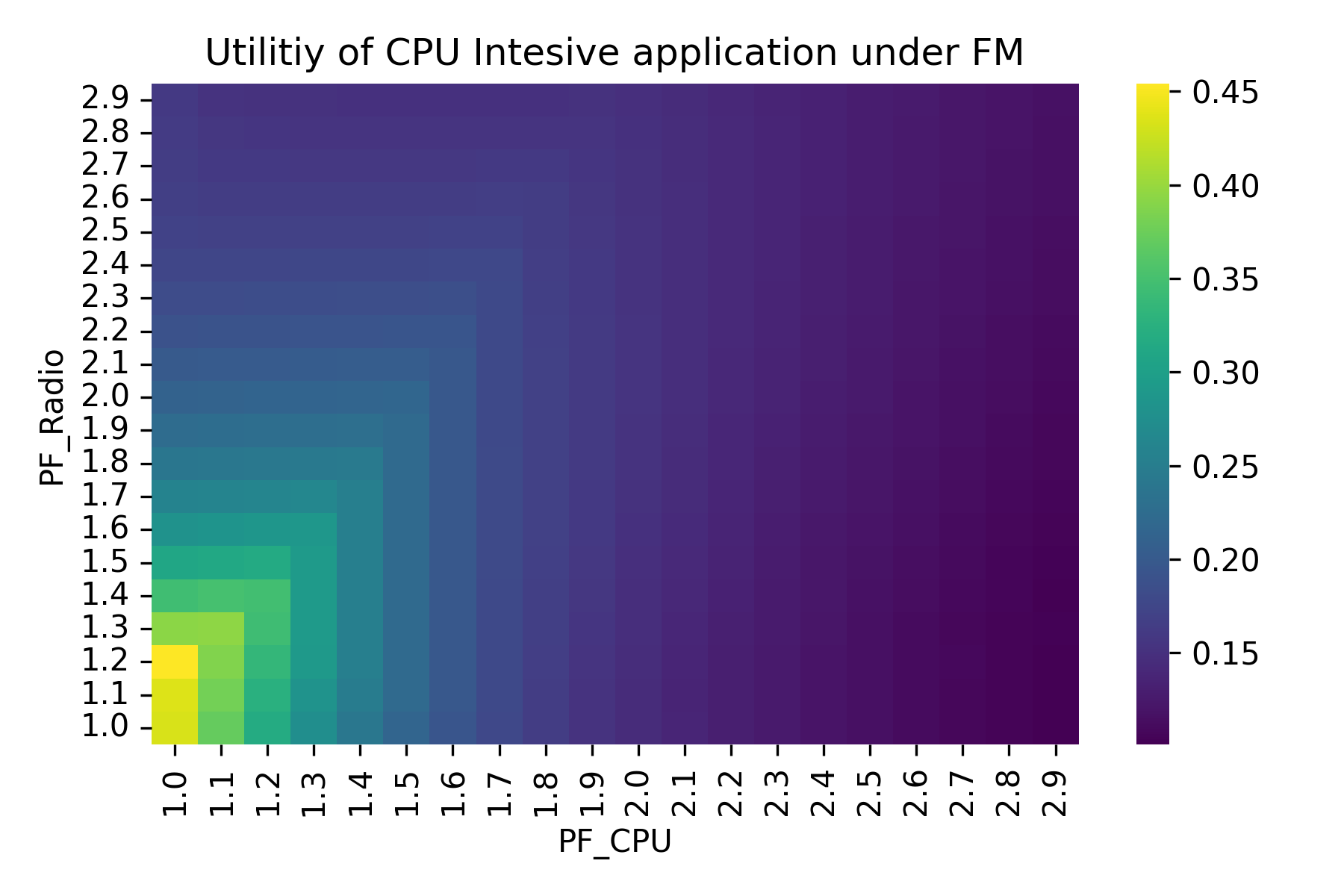} \captionsetup{skip=0pt}\caption{Fisher Market, CPU-intensive} \label{fig:fm_cpu} \end{subfigure} \hspace{-0.5cm}
\begin{subfigure}{0.42\textwidth} \centering \includegraphics[width=\linewidth]{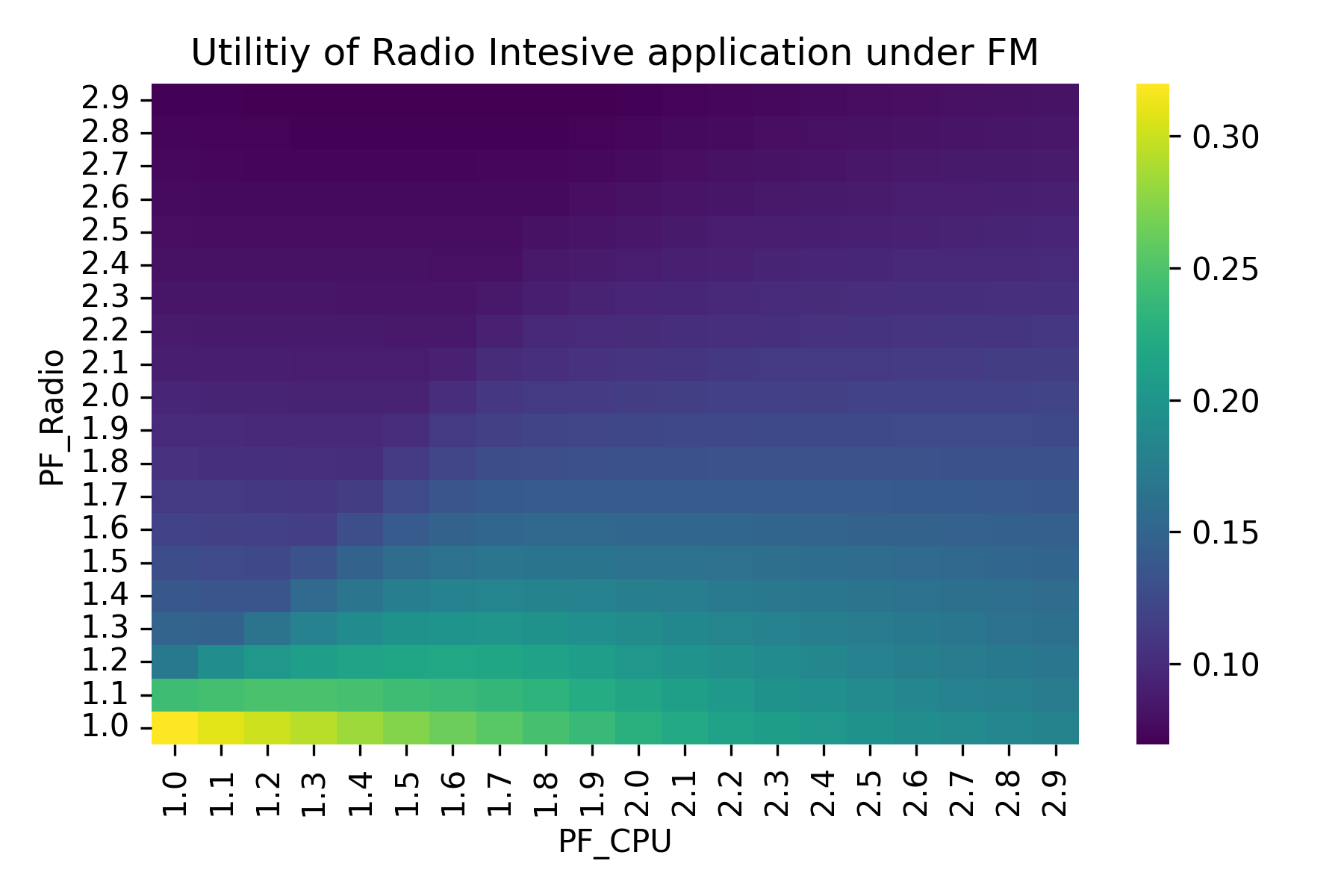} 
\captionsetup{skip=0pt}
\caption{Fisher Market, Radio-intensive} \label{fig:fm_radio} \end{subfigure}  \begin{subfigure}{0.42\textwidth} \centering \includegraphics[width=\linewidth]{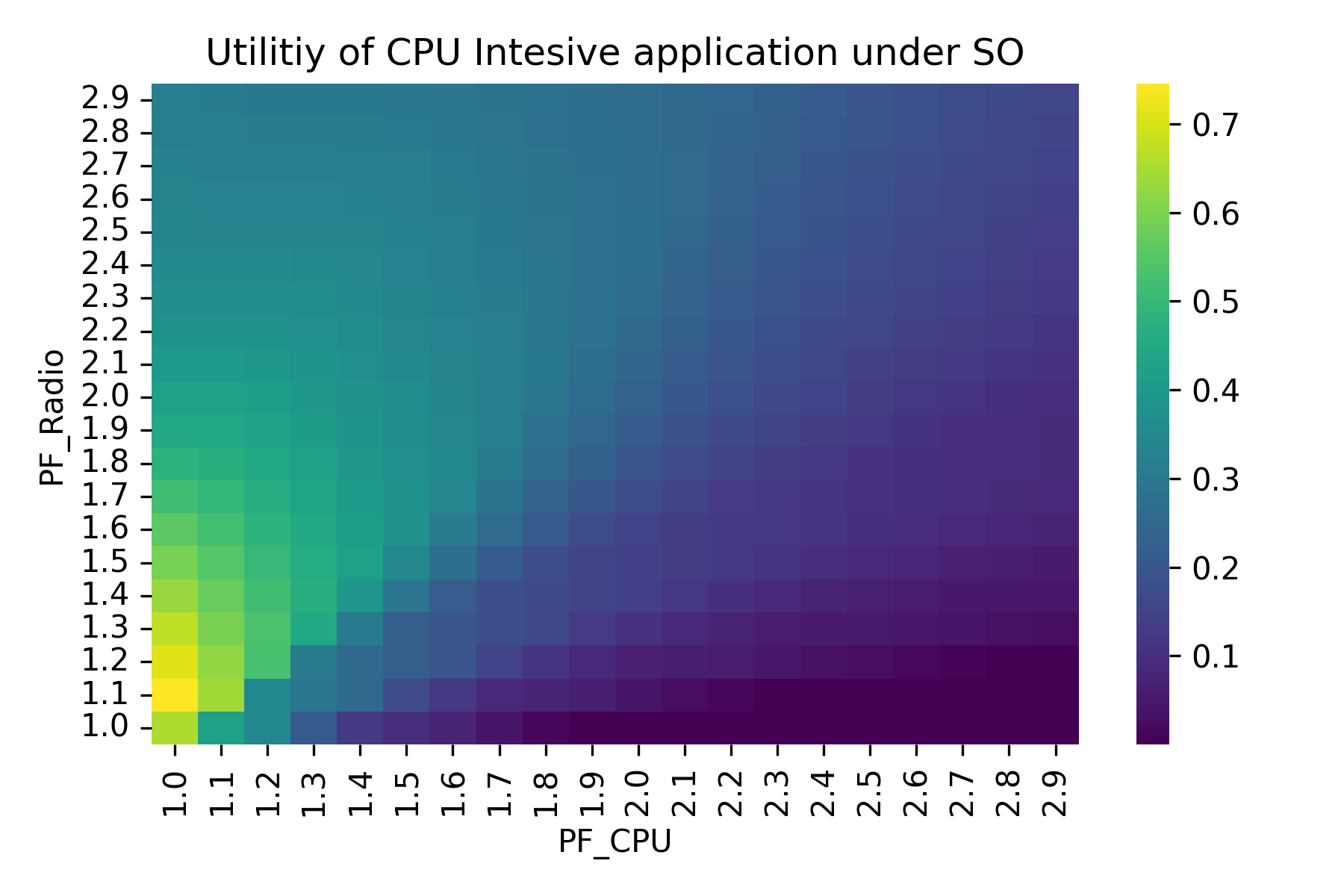}
\captionsetup{skip=0pt}
\caption{Social Optimum, CPU-intensive} \label{fig:so_cpu} \end{subfigure} \hspace{-0.5cm} \begin{subfigure}{0.42\textwidth} \centering \includegraphics[width=\linewidth]{Pictures/SO_heatmmap_CPU.png} 
\captionsetup{skip=0pt}
\caption{Social Optimum, Radio-intensive} \label{fig:so_radio} \end{subfigure} \caption{Utility variation of CPU-intensive (a, c) and radio-intensive (b, d) service providers as CPU (x-axis) and radio (y-axis) energy cost exponents vary, under Fisher Market and Social Optimum schemes.} \label{fig:multiple_images} \end{figure*}
Fig. \ref{fig:restriction}(b) shows that even though the restrictions are applied only at location 1, they also impact CPU usage at location 2. This happens because both locations share two common cloud facilities and are subject to the same global energy constraints. In particular, we observe that as energy constraints are relaxed, both SPs reduce their cloud usage. To compensate, the cloud-intensive SP increases its use of EC resources. The trends in Fig. \ref{fig:restriction}(b) and (c) are further supported by Fig. \ref{fig:restriction}(d), which shows how radio resource usage changes with varying restrictions. As restrictions at location 1 are relaxed, CPU-intensive SPs increasingly use radio to offload to the cloud, and similarly, radio-intensive SPs offload more to the EC. In contrast, the use of EC's CPU by CPU-intensive SPs decreases. Fig. 4 shows how the utility of SPs changes with the power $\beta$ in the energy consumption function, under both SO and FM allocation schemes. The heatmaps reveal that the FM-based scheme is less sensitive to changes in the cost function, resulting in a more balanced utility distribution between SPs. In contrast, the SO scheme is more sensitive, often allocating most of the resources to one SP, while giving almost nothing to the other.
\section{Conclusions and Future work}\label{sec:conclusions}
We proposed a resource allocation framework that takes into account heterogeneous resources, including radio, computing and energy (or, costs), with the latter varying spatially across different locations, e.g., due to different restrictions or energy mix. These heterogeneous resource and the interactions between SPs and InPs generates a complex multi-dimensional trade-off, where utilities and costs are to be balanced to promote fairness under different constraints. In this work, we intentionally consider generic functions linking resource usage and energy (or, carbon footprint), in order to explore the trade-offs in a technology-agnostic manner. Future works include specifying these functions under assumptions on exploited technologies and energy mix among others, to study the proposed trade-offs in real world applications and scenarios. 
\bibliographystyle{IEEEtran}
\bibliography{biblio_full}
\end{document}